\documentclass[10pt,final,journal,twocolumn]{IEEEtran}%
\usepackage{amsmath}
\usepackage{graphicx}
\usepackage{latexsym}
\usepackage{amssymb}
\usepackage{cite}
\usepackage{color}
\usepackage{multirow}
\usepackage{algorithmic,algorithm}
\usepackage{footnote}
\usepackage{flushend}

\newtheorem{lemma}{Lemma}

\begin{document}

\title{\LARGE Energy Harvesting Enabled MIMO Relaying through Time Switching}

\author{Jialing Liao, Muhammad R. A. Khandaker and Kai-Kit Wong%, Yangyang Zhang\\Zhongbin Zheng, and Christos Masouros
\thanks{Y. Zhang is with the Kuang-Chi Institute of Advanced Technology and Z. Zheng is with the East Institute of CATR of MIIT, China, while the other authors are with the Department of Electronic and Electrical Engineering, University College London, United Kingdom.}\thanks{This work is supported by EPSRC under grant EP/K015893/1.}}

\pagestyle{headings}
\maketitle \thispagestyle{empty}
%\vspace*{-2.5em}

\begin{abstract}
This paper considers simultaneous wireless information and power transfer (SWIPT) for a multiple-input multiple-output (MIMO) relay system. The relay is powered by harvesting energy from the source via time switching (TS) and utilizes the harvested energy to forward the information signal. Our aim is to maximize the rate of the system subject to the power constraints at both the source and relay nodes. In the first scenario in which the source covariance matrix is an identity matrix, we present the joint-optimal solution for relaying and the TS ratio in closed form. An iterative scheme is then proposed for jointly optimizing the source and relaying matrices and the TS ratio.
\end{abstract}

%\begin{keywords}
%Robust, secrecy, SWIPT.
%\end{keywords}

%\vspace{-.15in}
\section{Introduction}
Combining multiple-input multiple-output (MIMO) antenna and relaying is a promising means to enhance both coverage and performance of wireless communications networks \cite{ruhul11, ruhul12, ruhul13, ruhul14, ruhul15, ruhul16, ruhul_yue}. In \cite{swipt_1}, the optimal structure of the relay matrix to maximize the rate assuming unitary source precoding was presented. Joint source and relay optimization was considered in \cite{swipt_2, swipt_3}. Recently, more complex applications such as robust beamforming design with imperfect channel state information (CSI) were studied \cite{swipt_4}. On the other hand, energy harvesting \cite{ruhul_kk, ruhuul2} has emerged as an attractive component for relaying and cooperative communication as the battery storage of the relay nodes is usually limited.

Research has been carried out to study simultaneous wireless information and power transfer (SWIPT) for multi-antenna relay networks \cite{swipt_5, swipt_6, swipt_7}. For complexity reasons, preference was given to the time-switching (TS) mechanism over power splitting for energy harvesting, e.g., \cite{swipt_6, swipt_7}. Unfortunately, the existing approaches fail to provide a joint TS coefficient and precoding matrices design for a generic MIMO relay system. Also, \cite{swipt_5, swipt_6} depend on either semi-definite relaxation (SDR) and existing solvers or iterative approaches while in \cite{swipt_7} a closed form solution was given for a MISO relay system rather than general MIMO network. Closed-form solutions as well as the structures of the source covariance and relay beamforming matrices are not well understood.

Most recently, power-splitting based energy harvesting is considered for MIMO relay networks in \cite{jialing_ps}. In order to derive the maximum capacity with power constraints at the source and relay nodes, the authors jointly optimize the source covariance matrix, relay beamforming matrix and PS ratio. Instead of using SDR and software solvers, the structures of the optimal source covariance and relay precoding matrix are derived based on which iterative approaches are employed to derive the near-optimal results.

Unlike \cite{jialing_ps}, this paper considers the rate maximization problem for the MIMO relay system with an energy harvesting relay node employing TS. The fixed source covariance matrix case and the joint source, relay, and TS ratio optimization case are both investigated. Unlike the existing attempts which rely on SDR, we give the structures of the optimal relay beamforming matrix and the source covariance matrix and propose a closed-form solution and an iterative solution for the two cases, respectively.

\section{System Model}
We consider a two-hop MIMO relay network with an energy harvesting relay node and assume that the source, relay and destination nodes are all equipped with multiple antennas. The numbers of antennas are $M, L, N$, respectively. The relay harvests energy from the source and uses it to forward the information. Here we assume that the direct link between the source and destination is negligible with perfect channel state information (CSI) known at all nodes.

\begin{figure}
\centering
\includegraphics*[width=8.5cm]{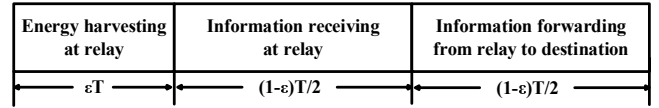}
\caption{The framework of the proposed TS relaying.}\label{TSR}
\end{figure}

The TS-based relaying involves three phases, as shown in Fig.~\ref{TSR}, with $T$ being the block length, and $\varepsilon$ denoting the TS ratio. In the first phase, the channel and source covariance matrices are defined as ${\bf \widetilde{H}}_1$ and ${\bf \widetilde{Q}}$, respectively. ${\bf \widetilde{s}}$ is the source symbol vector. In the information transmit phases, we will use ${\bf s}$ to denote the source symbol vector, $P$ to denote maximum transmit power, ${\bf H}_1$ and ${\bf H}_2$ to represent the channel matrices between the source and relay, and the relay and the destination, respectively, ${\bf Q}$ to denote the source covariance matrix and ${\bf F}$ the relay beamforming matrix. We will also consider the additive white Gaussian noises (AWGNs), ${\bf n}_1$ and ${\bf n}_2$, at the relay and destination nodes with variances $\sigma_1^2$ and $\sigma_2^2$, respectively.

\section{Relay and TS Ratio Only Design}
Here, we fix the source covariance matrix in the information transmit phase as ${\bf Q}=$$\mathbb{E}{({\bf s}{\bf s}^H)}$$=\frac {P}{D}{\bf I}$, where $D~(D \le \min(M, L, N))$ is the number of data streams. The superscript $H$ is the Hermitian operator and ${\bf I}$ is an identity matrix.

The harvested power at the relay can be expressed as
\begin{equation}
{\rm tr}{({\bf y}_e {\bf y}_e^H)} ={\rm tr}{({\bf \widetilde{H}}_1{\bf \widetilde{Q}}{\bf \widetilde{H}}_1^H+{\sigma_1^2}{\bf I}_D)},
\end{equation}
where $ {\bf y}_e = {\bf \widetilde{H}}_1 {\bf \widetilde{s}}+{\bf n}_1$ is the received signal at the relay in the energy harvesting phase.

Then in the information transmission phase, the received signal at the relay can be written as
\begin{equation}
{\bf y}_r = {\bf H}_1 {\bf s}+{\bf n}_1.
\end{equation}
In the last time phase, the received signal at the destination is
\begin{equation}
{\bf y}_d = {\bf H}_2 {\bf F} {\bf H}_1 {\bf s}+{\bf H}_2 {\bf F} {\bf n}_1+{\bf n}_2.
\end{equation}
As a result, the achievable rate is given by \cite{swipt_1}
\begin{equation}\label{cap}
C=\frac {1-\varepsilon}{2}\log_2 {\rm det}\left({\bf I}_D+\rho_1 {\bf H}_1{\bf H}_1^H -\rho_1 {\bf H}_1{\bf H}_1^H{\bf S}^{-1}\right),
\end{equation}
where ${\bf S}={\bf I}_D+\frac {\sigma_1^2}{\sigma_2^2} {\bf F}^H{\bf H}_2^H{\bf H}_2 {\bf F}$, and $T=1$ is assumed. The signal-to-noise ratio (SNR) at the relay is defined as $\rho_1\triangleq\frac {P}{D \sigma_1^2}$,

The transmitted signal at the relay
\begin{equation}
{\bf x}_r ={\bf F} {\bf H}_1 {\bf s}+{\bf F} {\bf n}_1
\end{equation}
will have to satisfy the harvested power constraint
\begin{equation}\label{tpc}
\frac {1-\varepsilon}{2}\sigma_1^2{\rm tr}{({\bf F} ({\bf I}_D+\rho_1 {\bf H}_1{\bf H}_1^H){\bf F}^H)} \le \varepsilon \eta {\rm tr}{({\bf \widetilde{H}}_1{\bf \widetilde{Q}}{\bf \widetilde{H}}_1^H+{\sigma_1^2}{\bf I}_D)}
\end{equation}
where $0 \le \eta \le 1 $ is the energy conversion efficiency.

Note that in the energy harvesting phase, the objective is to maximize the harvested power subject to the transmit power constraint, solution to which is given in \cite{swipt_8} as described below.
\begin{lemma}
Let the singular value decomposition (SVD) of ${\bf \widetilde{H}}_1$ be ${\bf \widetilde{H}}_1={\bf U}\Gamma^{\frac {1}{2}} {\bf V}^H$ where ${\bf U}$ and ${\bf V}$ are unified while the diagonal elements of $\Gamma$, $g_1, g_2,...,g_D$, are arranged in an descending order. $v_1$ is defined as the first column of ${\bf V}$. $P_0$ denotes the transmit power threshold in the energy harvesting phase. The optimal solution of the optimization problem to
\begin{align}
\max_{\bf \widetilde{Q}}\quad {\rm tr} ({\bf \widetilde{H}}_1 {\bf \widetilde{Q}}{\bf \widetilde{H}}_1^H+\sigma_1^2{\bf I}_D) &\\ \mbox{s.t.}\quad \quad{\rm tr} ({\bf \widetilde{Q}})\leq P_0, {\bf \widetilde{Q}}\succeq 0.
\end{align}
is ${\bf \widetilde{Q}}=P_0v_1v_1^H$ and the corresponding maximum harvested power is given by $g_1P_0+\sigma_1^2D$.
\end{lemma}
\begin{proof}
See \cite[Proposition 2.1]{swipt_8}.
\end{proof}
Let us now define ${\bf G}=\frac {\sigma_1}{\sigma_2}{\bf F}$ and formulate the following optimization problem
\begin{multline}\label{maxRob1}
\max_{{\bf G}, \varepsilon} ~~ C~~\mbox{s.t.}~~
\frac {1-\varepsilon}{2}{\rm tr}{({\bf G} ({\bf I}_D+\rho_1 {\bf H}_1{\bf H}_1^H){\bf G}^H)}\\
\le \frac {\varepsilon \eta}{\sigma_2^2} (g_1P_0+\sigma_1^2D).
\end{multline}
Note that for fixed ${\varepsilon}$, problem \eqref{maxRob1} becomes technically identical to the one considered in \cite{swipt_1}. Moreover, it can be proved that the presence of ${\varepsilon}$ does not change the optimal structure of the relay matrix. Hence, we have the relay matrix given by
\begin{equation}
{\bf F}={\bf V}_2 {\bf \Lambda }_F {\bf U}_1^H,
\end{equation}
where ${\bf \Lambda }_F$ denotes a diagonal matrix. $ {\bf V}_2 $ and $ {\bf U}_1$ come from the SVDs of the matrices:
\begin{align}
{\bf H}_1&={\bf U}_1 {\bf \Sigma}_1 {\bf V}_1^H,\label{channel} \\
{\bf H}_2&={\bf U}_2 {\bf \Sigma}_2 {\bf V}_2^H. \label{channel1}
\end{align}
%Here ${\bf \Sigma}_1, {\bf \Sigma}_2$ are diagonal matrices while others are unitary. The proof is similar to the relevant part in \cite{swipt_1} and hence ignored to avoid redundancy.

Now, let ${\bf G}={\bf V}_2 {\bf X}^\frac{1}{2} ({\bf I}+\rho_1 {\bf \Lambda}_1)^{-\frac {1}{2}} {\bf U}_1^H$, where ${\bf X}$ is a diagonal matrix with ${\bf X}={\rm diag}({x_1,x_2,\dots,x_D})$. In addition, we let ${\bf \Lambda_1}={\bf \Sigma}_1{\bf \Sigma}_1^H$ and ${\bf \Lambda_2}={\bf \Sigma}_2^H {\bf \Sigma}_2$ with the diagonal vectors $\boldsymbol{\alpha}=[\alpha_1,\dots,\alpha_D]$ and $\boldsymbol{\beta}=[\beta_1,\dots,\beta_D]$, respectively. The optimization problem (\ref {maxRob1}) can then be rewritten as
\begin{subequations}\label{maxRob2}
\begin{align}
&\max_{\{x_k\ge 0\}, 0\le \varepsilon\le 1} f(\{x_k\},\varepsilon)~~{\rm s.t.}\label{R2_o}\\
&g(\{x_k\},\varepsilon) \triangleq \frac {\varepsilon \eta}{\sigma_2^2} (g_1P_0+\sigma_1^2D)
-\frac {1-\varepsilon}{2}\sum_{k=1}^{D} x_k \ge 0,\label{R2_1}
\end{align}
\end{subequations}
in which
%\begin{multline}
$f(\{x_k\},\varepsilon) \triangleq\frac {1-\varepsilon}{2}
[\sum_{k=1}^{D}\log_2 (1+\rho_1\alpha_k)+\sum_{k=1}^{D}$ $\log_2(\frac {1+\beta_k x_k}{1+\rho_1\alpha_k+\beta_k x_k})]$.
%\end{multline}
Considering the Lagrangian, we have% the dual problem as
\begin{equation}\label{maxRob3}
\min_{\{x_k\ge 0\},0\le\varepsilon\le 1\atop \nu\ge 0,\{\lambda_k\ge 0\}} {\cal L} \triangleq f(\{x_k\},\varepsilon)+\nu g(\{x_k\},\varepsilon)+\sum_{k=1}^{D} \lambda_k x_k.% \\
%{\rm s.t.}~~0 \le \alpha \le 1, x_k \ge 0, \nu \ge 0, \lambda_k \ge 0, \forall k.
\end{equation}
Based on the Karush-Kuhn-Tucker (KKT) conditions, we have
\begin{subequations}\label{minRob4}
\begin{align}
\nu g(\{x_k\},\varepsilon) &=0, \label{R4_o}\\
\lambda_k x_k&=0, \forall k,\label{R4_1}\\
{\nabla_{x_k}}{\cal L}&=0, \forall k,\label{R4_2}\\
\nabla_\varepsilon{\cal L}&=0.\label{R4_3}
\end{align}
\end{subequations}
Since $\lambda_k \ge 0$, $k=1,...,D$, using (\ref {R4_2}) %, we conclude that
%\begin{equation}
%\frac {1-\alpha}{2\ln 2} [\frac {\beta_k}{1+\beta_k x_k}-\frac {\beta_k}{1+\rho_1\alpha_k+\beta_k x_k}]-\frac {1-\alpha}{2}\nu+\lambda_k =0.
%\end{equation}
we can show that
\begin{equation}
\nu \ge \frac{1}{\ln 2}  \frac {\left(\frac{\rho_1\alpha_k}{\beta_k}\right)}{\left(x_k+\frac {1}{\beta_k}\right)\left(x_k+\frac {\rho_1\alpha_k+1}{\beta_k}\right)}, \forall k.
\end{equation}
Let $h_k(x_k)=\frac{1}{\ln 2} \frac {\frac{\rho_1\alpha_k}{\beta_k}}{(x_k+\frac {1}{\beta_k})(x_k+\frac {\rho_1\alpha_k+1}{\beta_k})}$. It is obvious that $h_k$ decreases when $x_k$ rises from $0$ to $\infty$, i.e., $h_k(0)=\frac{1}{\ln 2} \frac {\rho_1\alpha_k\beta_k}{1+\rho_1\alpha_k}$, $h_k(\infty) \to 0$. Then according to (\ref {R4_1}), we know that
\begin{align}
x_k\left[\nu-\frac{1}{\ln 2}  \frac {\frac{\rho_1\alpha_k}{\beta_k}}{(x_k+\frac {1}{\beta_k})(x_k+\frac {\rho_1\alpha_k+1}{\beta_k})}\right]=0, \forall k.
\end{align}

If $\nu \ge \frac{1}{\ln 2} \frac {\rho_1\alpha_k\beta_k}{1+\rho_1\alpha_k}$, we know that $x_k=0$. Otherwise, if $0 < \nu < \frac{1}{\ln 2}  \frac {\rho_1\alpha_k\beta_k}{1+\rho_1\alpha_k}$, then we will have $x_k>0$ and $\nu = \frac{1}{\ln 2} \frac {\frac{\rho_1\alpha_k}{\beta_k}}{(x_k+\frac {1}{\beta_k})(x_k+\frac {\rho_1\alpha_k+1}{\beta_k})}$. Then according to \cite{swipt_1}, the optimal $x_k$ can be written as
\begin{align}\label{nu}
x_k=\frac{1}{2\beta_k}\left(\sqrt {\rho_1^2\alpha_k^2+\frac{4}{\ln 2} \rho_1\alpha_k\beta_k \mu}-\rho_1\alpha_k-2\right)^+,
\end{align}
where $(a)^+=\max\{0,a\}$ and $\mu=\frac{1}{\nu}$ can be obtained by substituting (\ref {nu}) into (\ref {R4_3}) such that
\begin{multline}
l(\mu)\triangleq-\frac{1}{2}\sum_{k=1}^{D}\log_2 \left[(1+\rho_1\alpha_k)\left(\frac {1+\beta_k x_k}{1+\rho_1\alpha_k+\beta_k x_k}\right)\right]\\
+\frac{1}{\mu}\left[\frac {\eta}{\sigma_2^2} (g_1P_0+\sigma_1^2D)+\frac {1}{2} \sum_{k=1}^{D} x_k\right]=0.
\end{multline}
Obviously, $l(\mu)$ decreases when $\mu \geq  \max_k\ln 2 \frac {1+\rho_1\alpha_k}{\rho_1\alpha_k\beta_k} $. Moreover, when $\mu \in [\min_k\ln 2 \frac {1+\rho_1\alpha_k}{\rho_1\alpha_k\beta_k},\max_k\ln 2 \frac {1+\rho_1\alpha_k}{\rho_1\alpha_k\beta_k}]$, $x_k$ either maintain at 0 or increases with $\mu$ which makes $l(\mu)$ decreases within this interval. To be exact, we have
\begin{equation}
l(\infty) \to -\frac {1}{2}\sum_{k=1}^{D} \log_2 (1+\rho_1\alpha_k) <0,
\end{equation}
and
\begin{multline}
l\left(\min_k\ln 2 \frac {1+\rho_1\alpha_k}{\rho_1\alpha_k\beta_k}\right)\\
=\left(\max_k\frac{1}{\ln 2} \frac {\rho_1\alpha_k\beta_k}{1+\rho_1\alpha_k}\right) \frac {\eta}{\sigma_2^2} (g_1P_0+\sigma_1^2D) >0.
\end{multline}
In contrast, when $ \mu \in (0, \min_k\ln 2 \frac {1+\rho_1\alpha_k}{\rho_1\alpha_k\beta_k}]$, $x_k=0, \forall k$ and
\begin{equation}
l(u) =\frac{1}{\mu} \frac {\eta}{\sigma_2^2} (g_1P_0+\sigma_1^2D)>0.
\end{equation}
As a consequence, we can always find an optimal value $\mu^* \in (\min_k\ln 2 \frac {1+\rho_1\alpha_k}{\rho_1\alpha_k\beta_k}, \infty)$ which makes $l(\mu^*)=0$ and this can be done by root-finding approaches such as bisection. When $\mu^*$ and $x_k, \forall k,$ are known, the optimal TS ratio can easily be calculated according to (\ref {R4_o}), which is given by
\begin{equation}
\varepsilon^* =\frac {\sum_{k=1}^{D} {x_k^*}}{\frac {2\eta}{\sigma_2^2} (g_1P_0+\sigma_1^2D)+\sum_{k=1}^{D} {x_k^*}}.
\end{equation}

\section{Joint Source, Relay and TS Ratio Design}
In this section, we consider a more general but challenging scenario with any available ${\bf Q}$. Based on (\ref {cap}), the instantaneous achievable rate in this case is given by
\begin{equation}
C %&\triangleq \frac {1-\varepsilon}{2}\log_2 |{\bf I}_L+\frac{{\bf H}_1 {\bf Q} {\bf H}_1^H}{\sigma_1^2} -\frac{{\bf H}_1 {\bf Q} {\bf H}_1^H}{\sigma_1^2({\bf I}_L+\frac {\sigma_1^2}{\sigma_2^2}{\bf F}^H{\bf H}_2^H{\bf H}_2 {\bf F})}|\notag\\
=\frac {1-\varepsilon}{2}\log_2 {\rm det}\left({\bf I}_D+\frac{{\bf H}_2 {\bf F}{\bf H}_1 {\bf Q} {\bf H}_1^H{\bf F}^H{\bf H}_2^H}{\sigma_2^2{\bf I}_D+\sigma_1^2{\bf H}_2 {\bf F}{\bf F}^H{\bf H}_2^H}\right).
\end{equation}
Then the optimization problem of interest becomes
\begin{multline}\label{maxRobC}
%\begin{align}
\max_{{\bf F},\varepsilon \atop {\rm tr}{({\bf Q})} \le P}C~~{\rm s.t.}~~
%&{\rm tr}{({\bf Q})} \le P,\\
\frac {1-\varepsilon}{2}{\rm tr}{(\sigma_1^2{\bf F}{\bf F}^H +{\bf F}{\bf H}_1 {\bf Q} {\bf H}_1^H{\bf F}^H)}\\
\le \varepsilon \eta {\rm tr}{({\bf \widetilde{H}}_1{\bf \widetilde{Q}}{\bf \widetilde{H}}_1^H+{\sigma_1^2}{\bf I}_D)}.
\end{multline}
By introducing an equivalent channel $\widehat{{\bf H}}_1={\bf H}_1 {\bf Q}^{\frac {1}{2}}$, the optimization problem becomes similar to the previous fixed source covariance matrix case. Therefore, we have $\widehat{{\bf F}}={\bf V}_2 \widehat{{\bf \Sigma}}_F \widehat{{\bf U}}_1^H$ where $\widehat{{\bf \Sigma}}_F$ is diagonal, and $\widehat{{\bf U}}_1$ and ${\bf V}_2$ come from the SVDs of $\widehat{{\bf H}}_1 = \widehat{{\bf U}}_1 \widehat{{\bf \Sigma}}_1 \widehat{{\bf V}}_1^H$, and ${\bf H}_2$ given in \eqref{channel1}.

As can be observed, both the objective function and the energy harvesting constraint have nothing to do with $\widehat{{\bf U}}_1$which indicates that any available $\widehat{{\bf H}}_1$ with the same $\widehat{{\bf \Sigma}}_1$ acts equally in terms of capacity and the energy harvesting constraint. That is to say, the optimal ${\bf Q}$ must require the least transmit power. Considering the fact that the presence of the TS ratio $\varepsilon$ will not change the structures of the source covariance and relay processing matrices, we provide the optimal structures of the source covariance and relay processing matrices below.

\begin{lemma}
The optimal solution of the optimization problem (\ref{maxRobC}) has the following structures
\begin{align}
{\bf F}&={\bf V}_2 {\bf \Sigma }_F {\bf U}_1^H,\label{structure1} \\
{\bf Q}&={\bf V}_1 {\bf \Lambda }_Q {\bf V}_1^H,\label{structure2}
\end{align}
where ${\bf \Sigma}_F, {\bf \Lambda}_Q$ are diagonal matrices, and the unitary matrices ${\bf U}_1, {\bf V}_1,{\bf U}_2, {\bf V}_2$ have been defined in (\ref{channel}) and (\ref{channel1}).
\end{lemma}

\begin{proof}
See \cite[Theorem 1]{swipt_2}.
\end{proof}

Then we let ${\bf \Lambda_Q}={\rm diag}({q_1,q_2,\dots,q_D})$, and ${\bf \Lambda_F}={\bf \Sigma}_F^2={\rm diag}({f_1,f_2,\dots,f_D})$. Substituting (\ref {structure1}) and (\ref {structure2}) into (\ref {maxRobC}) and introducing a set of new variables $d_k=f_k(\alpha_kq_k+{\bf \sigma}_1^2), \forall k$, the optimization problem (\ref{maxRobC}) can be rewritten as
\begin{subequations}\label{maxRobe}
\begin{align}
\max_{0\le \varepsilon\le 1,\{d_k\},\{q_k\}} &~~\tilde{f}(\varepsilon,\{d_k\},\{q_k\}) \label{Re_0} \\
{\rm s.t.} & ~~\sum_{k=1}^{D} q_k \le P,\label{Re_1} \\
&~~\tilde{g}(\varepsilon,\{d_k\},\{q_k\})\ge 0,\label{Re_2}
\end{align}
\end{subequations}
where we have defined
\begin{align}
\tilde{f}(\varepsilon, \{d_k\},\{q_k\}) &\triangleq \frac {1-\varepsilon}{2}\sum_{k=1}^{D}\log_2 \frac {\left(1+\frac{\alpha_k}{\sigma_1^2}q_k\right)\left(1+\frac{\beta_k}{\sigma_2^2}d_k\right)}{1+\frac{\alpha_k}{\sigma_1^2}q_k+\frac{\beta_k}{\sigma_2^2}d_k},\\
\tilde{g}(\varepsilon,\{d_k\},\{q_k\}) &\triangleq \varepsilon \eta (g_1P_0+\sigma_1^2D)-\frac {1-\varepsilon}{2}\sum_{k=1}^{D} d_k.
\end{align}

Note that (\ref{maxRobe}) involves only scalar variables in contrast to matrix variables in (\ref{maxRobC}). But the problem is still non-convex and a closed-form solution is difficult to obtain. In the following, we develop an alternating optimization based iterative algorithm which can be proved to converge at least to a local optimal solution. Since the subproblems are convex, close-form solutions are derived by solving Lagrangian dual problems. To proceed, we let $\boldsymbol{q}=[q_1,q_2,\dots,q_D]^T$, and $\boldsymbol{d}=[d_1, d_2,\dots,d_D]^T$.  % The iterative procedure completes in two phases.

\subsection{Optimization with fixed $\boldsymbol{q}$}
We first fix $\boldsymbol{q}$ and search for the optimal $\boldsymbol{d}$ and $\varepsilon$ with the given $\boldsymbol{q}$. %The optimization problem can be expressed as
%\begin{subequations}\label{maxRobe1}
%\begin{align}
%\max_{0\le \varepsilon\le 1,\{d_k\}}~~\tilde{f}(\varepsilon, \{d_k\}) ~~
%{\rm s.t.}~~\tilde{g}(\varepsilon,\{d_k\})\ge 0.
%\end{align}
%\end{subequations}
Considering the Lagrangian of the problem, we have the following dual problem:
\begin{equation}\label{maxRobea}
\max_{0\le\varepsilon\le 1,\{d_k\ge 0\}\atop \nu_1\ge 0,\{\lambda_k\ge 0\}} {\cal L} \triangleq \tilde{f}(\varepsilon,\{d_k\})+\nu_1 \tilde{g}(\varepsilon,\{d_k\})+\sum_{k=1}^{D} \lambda_k d_k.
%{\rm s.t.}~~0 \le \varepsilon \le 1, d_k \ge 0, \nu \ge 0, \lambda_k \ge 0, \forall k.
\end{equation}
Based on the KKT conditions, we have
\begin{subequations}\label{minRobf}
\begin{align}
\nu_1 \tilde{g}(\varepsilon,\{d_k\})& =0,\label{Rf_o}\\
\lambda_k d_k&=0,\forall k, \label{Rf_1}\\
{\nabla_{d_k}}{\cal L}&=0, \forall k, \label{Rf_2}\\
\nabla_\varepsilon{\cal L}&=0. \label{Rf_3}
\end{align}
\end{subequations}

Comparing with the fixed source covariance matrix case, we notice that $d_k$ here is equivalent to $\sigma_2^2 x_k$. As a result, we have
%Comparing with the fixed source covariance matrix case mentioned in the previous section, it is easy to find that $d_k$ here is equivalent to $\sigma_2^2 x_k$. As a result, we have
\begin{equation}\label{expd}
d_k=\frac{\sigma_2^2}{2\beta_k}\left(\sqrt {\left(\frac{q_k}{\sigma_1^2}\alpha_k\right)^2+\frac{4q_k\alpha_k\beta_k \mu_1}{\sigma_1^2\sigma_2^2\ln 2}}-\frac{q_k}{\sigma_1^2}\alpha_k-2\right)^+.
\end{equation}
Meanwhile, $\mu_1=\frac{1}{\nu_1}$ is decided by (\ref {Rf_3}). Substituting (\ref {expd}) into (\ref {Rf_3}), we have
\begin{multline}\label{lmu}
\widetilde{l}(\mu_1)\triangleq-\frac{1}{2}\left[\sum_{k=1}^{D}\log_2 \frac{\left(1+\frac{\alpha_k}{\sigma_1^2}q_k\right)\left(1+\frac{\beta_k}{\sigma_2^2}d_k\right)}{1+\frac{\alpha_k}{\sigma_1^2}q_k+\frac{\beta_k}{\sigma_2^2}d_k}\right]\\
+\frac{1}{\mu_1}\left[\eta (g_1P_0+\sigma_1^2D)+\frac {1}{2} \sum_{k=1}^{D} d_k\right]=0.
\end{multline}
Using \eqref{lmu}, we can obtain the optimal $\mu_1$ by a bisection search and then use it calculate $\boldsymbol{d}$. Then according to (\ref {Rf_o}), we have
\begin{align}\label{tsratio}
\varepsilon^* =\frac {\sum_{k=1}^{D} {d_k}}{2\eta (g_1P_0+\sigma_1^2D)+\sum_{k=1}^{D} {d_k}}.
\end{align}

\subsection{Optimization with fixed $\boldsymbol{d}$ and $\varepsilon$}
Since the transmit power constraint at relay (\ref {Re_2}) is independent on ${\bf q}$, and thus can be satisfied with the obtained $\boldsymbol{d}$ and $\varepsilon$ when fixing $\boldsymbol{q}$, here we do not need to consider it any longer. %Thus we have
%\begin{subequations}\label{maxRobeb}
%\begin{align}
%\max_{\{q_k\}} ~~\tilde{f}(\{q_k\}) ~~~~{\rm s.t.}~~\sum_{k=1}^{D} q_k \le P.
%\end{align}
%\end{subequations}
Referring to the Lagrangian of (\ref {maxRobe}), we have
\begin{multline}\label{maxRobg}
\max_{\{q_k\ge 0\},\atop \nu_2\ge 0, \{\lambda_k\ge 0\}} {\cal L}\triangleq \tilde{f}(\{q_k\})+\nu_2 \left(P-\sum_{k=1}^{D} q_k\right)
+\sum_{k=1}^{D} {\lambda_k q_k}.
%{\rm s.t.}~~&0 \le \varepsilon \le 1, \nu_1 \ge 0,\nu_2 \ge 0, \\
%&q_k \ge 0, \lambda_k \ge 0, \forall k.
\end{multline}
Deriving the relevant KKT conditions again, we obtain
\begin{subequations}\label{minRobh}
\begin{align}
\nu_2 \left(P-\sum_{k=1}^{D} q_k\right) &=0,\label{Rhh_0}\\
\lambda_k q_k&=0, \forall k,\label{Rh_2}\\
{\nabla_{q_k}}{\cal L}&=0,\forall k.\label{Rhh_3}
\end{align}
\end{subequations}
Then according to (\ref {Rhh_3}), we have
\begin{multline}
\frac{1-\varepsilon}{2}\frac{\alpha_k}{\ln 2 \sigma_1^2}\left(\frac{1}{1+\frac{\alpha_k}{\sigma_1^2}q_k}-\frac{1}{1+\frac{\alpha_k}{\sigma_1^2}q_k+\frac{\beta_k}{\sigma_2^2}d_k}\right)\\
-\nu_2+\lambda_k=0. \label{eq_qk}
\end{multline}
Thus from \eqref{eq_qk}, we have%$q_k$ can be written as\frac{1-\varepsilon}{2}
\begin{equation}\label{expq}
q_k=\frac{\sigma_1^2}{2\alpha_k}\left[\sqrt {\left(\frac{\beta_k}{\sigma_2^2}d_k\right)^2+\frac{2(1-\varepsilon)\beta_k d_k \alpha_k \mu_2}{\sigma_1^2 \sigma_2^2\ln 2}}-\frac{\beta_k}{\sigma_2^2}d_k-2\right]^+,
\end{equation}
where $\mu_2=\frac{1}{\nu_2}$. Substituting (\ref {expq}) into (\ref {Rhh_0}), the optimal $v_2$ can easily be derived through root finding methods.

\subsection{Iterative Optimization}
The iteration to solve (\ref{maxRobC}) is given in Algorithm~1.

\begin{algorithm} 
    \caption{Iteration Framework for TS Relaying}
  \begin{algorithmic}[1]
	  \STATE \textbf{Initialization} Let $\boldsymbol{q}$ satisfying (\ref {Re_1})
		\STATE  Calculate optimal $\boldsymbol{d}$ and $\varepsilon$ with fixed $\boldsymbol{q}$ using (\ref {expd}) and (\ref {tsratio})
		\STATE  Re-optimize $\boldsymbol{q}$ with the obtained $\boldsymbol{d}$ and $\varepsilon$ using (\ref {expq})
    \STATE  Return to Step 2 until convergence
  \end{algorithmic}
\end{algorithm}
\subsection{Convergence Analysis}
Since both the subproblems are convex, the obtained closed form solution to each subproblem is optimal. Thus the objective function (\ref {Re_0}) will increase (or at least maintain) after each iteration. In addition, the constraints in the subproblems are always satisfied with equality in every conditional update. Therefore, the proposed iterative algorithm ensure a monotonic convergence towards (at least) a locally optimal solution.

\section{Simulation Results}
This section investigates the performance of the proposed schemes for the MIMO relay system. The results of the naive amplify-and-forward (NAF) algorithm are also provided for comparison. In the NAF algorithm, we use the $\varepsilon$ derived with the uniform source precoding scheme and let ${\bf Q} = \frac{P}{D}{\bf I}$ and ${\bf F} = \sqrt{\chi}{\bf I}$ where $\chi$ is the scalar that makes the constraint (\ref {tpc}) satisfied. Here we assume that $N=M=L$. Both of ${\bf H_1}$ and ${\bf H_2}$ are modeled as Rician fading channels with a series of independent zero-mean complex Gaussian random variables with variance of $-10 {\rm dB}$.

Fig.~\ref{simu1} plots the achievable rate of the proposed TS relaying schemes against various $P_0$ with $P=1$. The considered values of $P_0$ range from 0 dBm to 50 dBm. The numbers of antennas are all set to be $4$ and $6$, respectively. As is expected, the NAF algorithm falls far behind the proposed schemes. The joint source, relay and TS ratio optimization outperforms the relay and TS ratio only optimization with their gap increasing slightly as the numbers of antennas increase. In all cases, the achievable rate increases as either $P_0$ or the numbers of antennas increase. Fig.~\ref{simu2} then presents the rate results for different numbers of antennas with $P_0=P=1$. It is clear that the achievable rates of the proposed schemes are much better than that of the NAF algorithm and the joint source, relay and TS ratio optimization shows performance gain over the uniform source precoding scheme. Notably, if numbers of the antennas increase, the performance gaps among the three schemes also increase, which agrees with the results in Fig.~\ref{simu1}.

\begin{figure}
\centering
\includegraphics*[width=5.5cm]{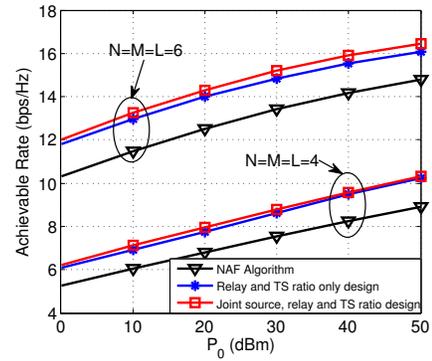}
\caption{Rate results against different $P_0$.}
\label{simu1}
\end{figure}

%\begin{figure}[ht]
%\centering
%\includegraphics*[width=5cm]{tsr_a.eps}
%\caption{Optimal TS ratio for TSR schemes against different noise variance $\sigma_1^2$ at relay and $\sigma_2^2=0{\rm dBm}$.}
%\label{simu2}
%\end{figure}

\begin{figure}
\centering
\includegraphics*[width=5.5cm]{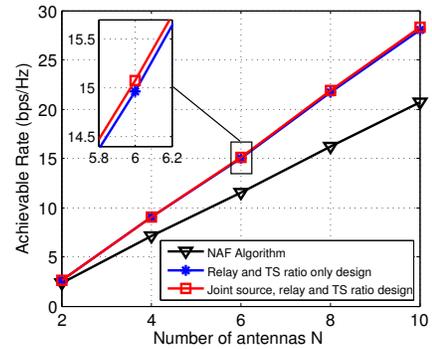}
\caption{Rate results against the number of antennas $N$.}
\label{simu2}
\end{figure}

\section{Conclusion}
This letter studied the rate maximization of a MIMO relay network with an energy harvesting relay node. We started with the fixed source covariance matrix scenario assuming uniform source precoding and then considered joint optimization with the source covariance. Closed-form solution as well as iterative scheme were proposed, respectively, for the two cases. %Simulations demonstrate that joint optimization of the source, relay and TS ratio yields capacity gain over the relay and TS ratio only optimization.

\bibliographystyle{IEEEtran}\footnotesize
%\IEEEtriggeratref{7}
%\bibliography{swipt}

% Generated by IEEEtran.bst, version: 1.13 (2008/09/30)

\end{document}